\documentclass[11pt, hidelinks,oneside]{article}   
\usepackage[utf8]{inputenc}
\usepackage{amsmath,amssymb,amsthm,algorithmicx,algorithm,hyperref}
\usepackage{graphicx,bbm}
\usepackage[top=1in,margin=1.2in]{geometry}
\usepackage[authoryear, sort]{natbib}
\usepackage{authblk,booktabs,epstopdf,color}
\usepackage[space]{grffile}
\date{}
\def\Gauss{{\mathrm{N}}}

\def\mr{\mathrm}

\def\LS{{\textbf{LASSO}}}
\def\HS{{\textbf{HS}}}
\def\T{\mathrm{\scriptscriptstyle{T}}}

\newtheorem{theorem}{Theorem}[section]

\newtheorem{proposition}[theorem]{Proposition}

\def\Gauss{{\mathrm{N}}}

\def\T{\mathrm{\scriptscriptstyle{T}}}

\def\mr{\mathrm}

\def\LS{{LASSO}}
\def\HS{{HS}}
\newcommand*{\email}[1]{\href{mailto:#1}{\nolinkurl{#1}} }

\begin{document}
\title{Fast sampling with Gaussian scale-mixture priors in high-dimensional regression}
\author{Anirban Bhattacharya \\ \href{mailto:anirbanb@stat.tamu.edu}{anirbanb@stat.tamu.edu} 
\and Antik Chakraborty  \\ \href{mailto:antik@stat.tamu.edu}{antik@stat.tamu.edu}, 
\and Bani K. Mallick  \\ \href{mailto:bmallick@stat.tamu.edu}{bmallick@stat.tamu.edu}}

\baselineskip=18pt

\maketitle
\begin{center}
Department of Statistics

Texas A\&M University

College Station, Texas, USA
\end{center}
\begin{abstract}
\baselineskip=16pt
We propose an efficient way to sample from a class of structured multivariate Gaussian distributions which routinely arise as conditional posteriors 
of model parameters that are assigned a conditionally Gaussian prior. The proposed algorithm only requires matrix operations in the form of matrix 
multiplications and linear system solutions. We exhibit that the computational complexity of the proposed algorithm grows linearly with the dimension 
unlike existing algorithms relying on Cholesky factorizations with cubic orders of complexity. The algorithm should be broadly applicable in settings 
where Gaussian scale mixture priors are used on high dimensional model parameters. We provide an illustration through posterior sampling in a high 
dimensional regression setting with a horseshoe prior \citep{carvalho2010horseshoe} on the vector of regression coefficients. 
\end{abstract}
\vspace{30mm}


\pagebreak
\section{Introduction}
Continuous shrinkage priors have recently received significant attention as a mechanism to induce approximate sparsity in high-dimensional 
parameters. Such prior distributions can mostly be expressed as global-local scale mixtures of Gaussians \citep{polson2010shrink,Bhattacharya_2014}.
These global-local priors \citep{polson2010shrink} aim to shrink noise coefficients while retaining any signal, 
thereby providing an approximation to the operating characteristics of discrete mixture priors \citep{george1997approaches,scott2010bayes}, which allow a subset of the parameters to be exactly zero. 

A major attraction of global-local priors has been computational efficiency and simplicity. Posterior inference poses a stiff challenge for 
discrete mixture priors in moderate to high-dimensional settings,  but the scale-mixture 
representation of global-local priors allows parameters to be updated in blocks via a fairly automatic Gibbs sampler in a wide variety of problems. These include regression \citep{caron2008sparse,armagan2011generalized}, variable selection \citep{hahn2015decoupling}, wavelet denoising \citep{polson2010shrink}, factor models and covariance 
estimation \citep{bhattacharya2011sparse,debdeep2013}, and time series \citep{durante2014locally}. Rapid mixing and convergence of the resulting Gibbs sampler for specific classes of priors has been recently established in the high-dimensional regression context by \cite{khare2013geometric} and \cite{pal2014geometric}.  Moreover, recent results suggest that a subclass of global-local priors can 
achieve the same minimax rates of posterior concentration as the discrete mixture priors in high-dimensional estimation problems \citep{Bhattacharya_2014,van2014horseshoe,debdeep2013}. 

In this article, we focus on computational aspects of global-local priors in the high-dimensional linear regression setting 
\begin{align}\label{eq:lin_reg}
y = X \beta + \epsilon, \quad \epsilon \sim \mbox{N}(0, \sigma^2 \mr I_n),
\end{align}
where $X \in \Re^{n \times p}$ is a $n \times p$ matrix of covariates with the number of variables $p$ potentially much larger than the sample size $n$. 
A global-local prior on $\beta$ assumes 
that
\begin{align}
&\beta_j \mid \lambda_j, \tau, \sigma \sim \mbox{N}(0, \lambda_j^2 \tau^2 \sigma^2), \quad (j = 1, \ldots, p), \label{eq:gl_1} \\
& \lambda_j \sim f, \quad (j =1, \ldots, p) \label{eq:gl_2} \\
& \tau \sim g, \quad \sigma \sim h, \label{eq:gl_3}
\end{align}
where $f, g$ and $h$ are densities supported on $(0, \infty)$. The $\lambda_j$s are usually referred to as local scale parameters 
while $\tau$ is a global scale parameter. Different choices of $f$ and $g$ lead to different classes of priors. For instance, a half-Cauchy 
distribution for $f$ and $g$ leads to the horseshoe prior of \citet{carvalho2010horseshoe}. In the $p \gg n$ setting where most entries of $\beta$ are assumed to be zero or close to zero, the choices of $f$ and $g$ play a key role in controlling the effective sparsity and concentration of the prior and posterior \citep{polson2010shrink,debdeep2013}.


Exploiting the scale-mixture representation \ref{eq:gl_1}, it is straightforward in principle to formulate a Gibbs sampler. 
The conditional posterior of $\beta$ given $\lambda = (\lambda_1, \ldots, \lambda_p)^{\T}, \tau$ and $\sigma$ is given by 
\begin{align}\label{eq:beta_post}
\beta \mid y, \lambda, \tau, \sigma \sim \Gauss(A^{-1} X^{\T} y,\sigma^2 A^{-1}), \quad A = (X^{\T} X +\Lambda_*^{-1}), 
\quad \Lambda_*=\tau^2 \mbox{diag}(\lambda_1^2,...,\lambda_p^2).
\end{align}
Further, the $p$ local scale parameters $\lambda_j$ have conditionally independent posteriors and hence $\lambda = (\lambda_1, \ldots, \lambda_p)^{\T}$
can be updated in a block by slice sampling
\citep{polson2014bayesian} if conditional posteriors are unavailable in closed form. However, unless care is exercised, sampling from \eqref{eq:beta_post} can be expensive for large values of $p$. Existing algorithms \citep{rue2001fast} to sample from \eqref{eq:beta_post} face a bottleneck for large $p$ to perform a Cholesky decomposition of $A$ at each iteration. One cannot resort to 
precomputing the Cholesky factors since the matrix $\Lambda^*$ in \eqref{eq:beta_post} changes from at each iteration.
%
In this article, 
we present an exact sampling algorithm for Gaussian distributions \eqref{eq:beta_post} which relies on data augmentation. We show that the computational complexity of the algorithm scales linearly in $p$. 

\section{The algorithm}\label{sec:gen}

Suppose we aim to sample from $\Gauss_p(\mu,\Sigma)$, with 
\begin{align}\label{eq:gen_setting}
\Sigma=(\Phi^{\T}\Phi+D^{-1})^{-1}, \quad \mu=\Sigma \Phi^{\T}\alpha,
\end{align}
where $D \in \Re^{p \times p}$ is symmetric positive definite, $\Phi \in \Re^{n \times p}$, and $\alpha \in \Re^{n \times 1}$; \eqref{eq:beta_post} is a special case of \eqref{eq:gen_setting} with $\Phi = X/\sigma$, $D = \sigma^2 \Lambda_*$ and $\alpha = 
y/\sigma$. A similar sampling problem arises in all the applications mentioned in \S 1, and the proposed approach can be used in such settings. In the sequel, we do not require $D$ to be diagonal, however we assume that $D^{-1}$ is easy to calculate and it is 
straightforward to sample from $\Gauss(0, D)$. This is the case, for example, if $D$ corresponds to the covariance matrix of an $\mbox{AR}(q)$ 
process or a Gaussian Markov random field.

Letting $Q = \Sigma^{-1} = (\Phi^{\T} \Phi + D^{-1})$ denote the precision, inverse covariance, matrix and $b = \Phi^{\T} \alpha$, we can 
write $\mu = Q^{-1} b$. \cite{rue2001fast} proposed an efficient algorithm to sample from a $\mbox{N}(Q^{-1}b, Q^{-1})$ distribution that avoids explicitly calculating the inverse of $Q$, which is computationally expensive and numerically unstable. Instead, the 
algorithm in \S 3.1.2. of \cite{rue2001fast} performs a Cholesky decomposition of $Q$ and uses the Cholesky factor to solve a series of linear 
systems to arrive at a sample from the desired Gaussian distribution.  
The original motivation was to efficiently sample from Gaussian Markov random fields where $Q$ 
has a banded structure so that the Cholesky factor and the subsequent linear system solvers can be computed efficiently. Since $Q = (\Phi^{\T} \Phi + D^{-1})$ does not have any special structure in the present setting, the Cholesky factorization has complexity $O(p^3)$; see \S 4.2.3 of \cite{golub2012matrix}, which increasingly becomes prohibitive for large $p$.  
%
We present an 
alternative exact mechanism to sample from a Gaussian distribution with parameters as in \eqref{eq:gen_setting} below:
\begin{algorithm}[h!]
\caption{ {\bf Proposed algorithm} } \label{alg:al2}
\begin{tabbing}
   \enspace (i) Sample $u \sim \Gauss(0,D)$ and $\delta \sim \Gauss(0,\mr I_n)$ independently. \\
   \enspace (ii) Set $v=\Phi u+\delta$.\\
   \enspace (iii) Solve $(\Phi D \Phi^{\T}+\mr I_n)w = (\alpha - v)$.\\
   \enspace (iv) Set $\theta=u+D \Phi^{\T}w$.
\end{tabbing}
\vspace*{-12pt}
\end{algorithm}
\begin{proposition}\label{prop:samp}
Suppose $\theta$ is obtained by following Algorithm \ref{alg:al2}. Then, $\theta \sim \Gauss(\mu, \Sigma)$, where $\mu$ and $\Sigma$ 
are as in \eqref{eq:gen_setting}. 
\end{proposition}
\begin{proof}
By the Sherman--Morrison--Woodbury identity \citep{hager1989updating} and some algebra, $\mu=D \Phi^{\T}(\Phi D \Phi^{\T}+\mr I_n)^{-1}\alpha$. By construction, $v \sim \Gauss(0,\Phi D \Phi^{\T}+\mr I_n)$. 
Combining steps (iii) and (iv) of Algorithm \ref{alg:al2}, we have $\theta=u+D\Phi^{\T}(\Phi D \Phi^{\T}+\mr I_n)^{-1}(\alpha-v)$. Hence $\theta$ has a normal distribution with mean $D \Phi^{\T}(\Phi D \Phi^{\T}+\mr I_n)^{-1}\alpha=\mu$. Since $\mathrm{cov}(u,v)=D\Phi^{\T}$, we obtain $\mathrm{cov}(\theta)=D-D\Phi^{\T}(\Phi D \Phi^{\T}+\mr I_n)^{-1}\Phi D=\Sigma$, again by the Sherman--Morrison--Woodbury identity. This completes the proof; a constructive proof is given in the Appendix.
\end{proof}
While Algorithm \ref{alg:al2} is valid for all $n$ and $p$, the 
computational gains are biggest when $p \gg n$ and $\mbox{N}(0, D)$ is easily sampled. 
Indeed, the primary motivation is to use data augmentation to cheaply sample $\zeta = (v^{\T}, u^{\T})^{\T} \in \Re^{n+p}$ and obtain a desired sample from \eqref{eq:gen_setting} via linear transformations and marginalization. When $D$ is diagonal, as in the case of global-local priors \eqref{eq:gl_1}, the complexity of the proposed algorithm
is $O(n^2 p)$; the proof uses standard results about complexity of matrix multiplications and linear system solutions;  see \S 1.3 \& 3.2 of \cite{golub2012matrix}. For non-sparse $D$, calculating $D \Phi^{\T}$ has a worst-case complexity of $O(np^2)$, which is the dominating term in the complexity 
calculations. 
In comparison to the $O(p^3)$ complexity of the competing algorithm in \cite{rue2001fast}, Algorithm \ref{alg:al2} therefore offers huge gains when $p \gg n$. For example, to run 6000 iterations of a Gibbs sampler for the horseshoe prior \citep{carvalho2010horseshoe} with sample size $n = 100$ in MATLAB on a INTEL(E5-2690) 2.9 GHz machine with 64 GB DDR3 memory, Algorithm \ref{alg:al2} takes roughly the same time as \cite{rue2001fast} when $p = 500$ but offers a speed-up factor of over 250 when $p = 5000$. MATLAB code for the above comparison and subsequent simulations is available at \url{ https://github.com/antik015/Fast-Sampling-of-Gaussian-Posteriors.git}. 

The first line of the proof implies that Algorithm \ref{alg:al2} outputs $\mu$ if one sets $u = 0, \delta = 0$ in step (i). 
The proof also indicates that the log-density of \eqref{eq:gen_setting} can be efficiently calculated at any $x \in \Re^p$. Indeed, since $\Sigma^{-1}$ is readily available, $x^{\T} \Sigma^{-1} x$ and $x^{\T} \Sigma^{-1} \mu$ are cheaply calculated and $\log |\Sigma^{-1}|$ can be calculated in $O(n^3)$ steps using the identity $|\mathrm{I}_r+AB|=|\mathrm{I}_s+BA|$ for $A \in \Re^{r \times s}, B \in \Re^{s \times r}$. Finally, from the proof, $\mu^\T \Sigma^{-1} \mu= \alpha^{\T}\Phi \Sigma^{-1} \Phi^{\T}\alpha=\alpha^{\T}\Phi \{D-D\Phi^{\T}(\Phi D \Phi^{\T}+\mathrm{I}_n)^{-1}\Phi D\}\Phi^{\T}\alpha=\alpha^{\T}\{ (\Phi D \Phi^{\T}+\mathrm{I}_n)^{-1}\Phi D \Phi^{\T} \} \alpha$, which can be calculated in $O(n^3)$ operations.

\section{Frequentist operating characteristics in high dimensions}\label{sec:app}

The proposed algorithm provides an opportunity to compare the frequentist operating characteristics of shrinkage priors in high-dimensional regression problems. We compare various aspects of the horseshoe prior \citep{carvalho2010horseshoe} to frequentist procedures and obtain highly promising results. 
We expect similar results for the Dirichlet--Laplace \citep{Bhattacharya_2014}, normal-gamma \citep{griffin2010inference} and generalized double-Pareto \citep{armagan2011generalized} priors, which we hope to report elsewhere.

\begin{figure}
\caption{Boxplots of $\ell_1$, $\ell_2$ and prediction error across 100 simulation replicates. $\mbox{HS}_{\mbox{me}}$ and $\mbox{HS}_{\mbox{m}}$ respectively denote posterior point wise median and mean for the horeshoe prior. True $\beta_0$ is $5$-sparse with non-zero entries $\pm \{1.5,1.75,2,2.25,2.5\}$. Top row: $\Sigma = \mathrm{I}_p$ (independent). Bottom row: $\Sigma_{jj}=1,\Sigma_{jj'}=0.5, j\ne j'$ (compound symmetry). }
\begin{flushleft}
\includegraphics[width=1\textwidth]{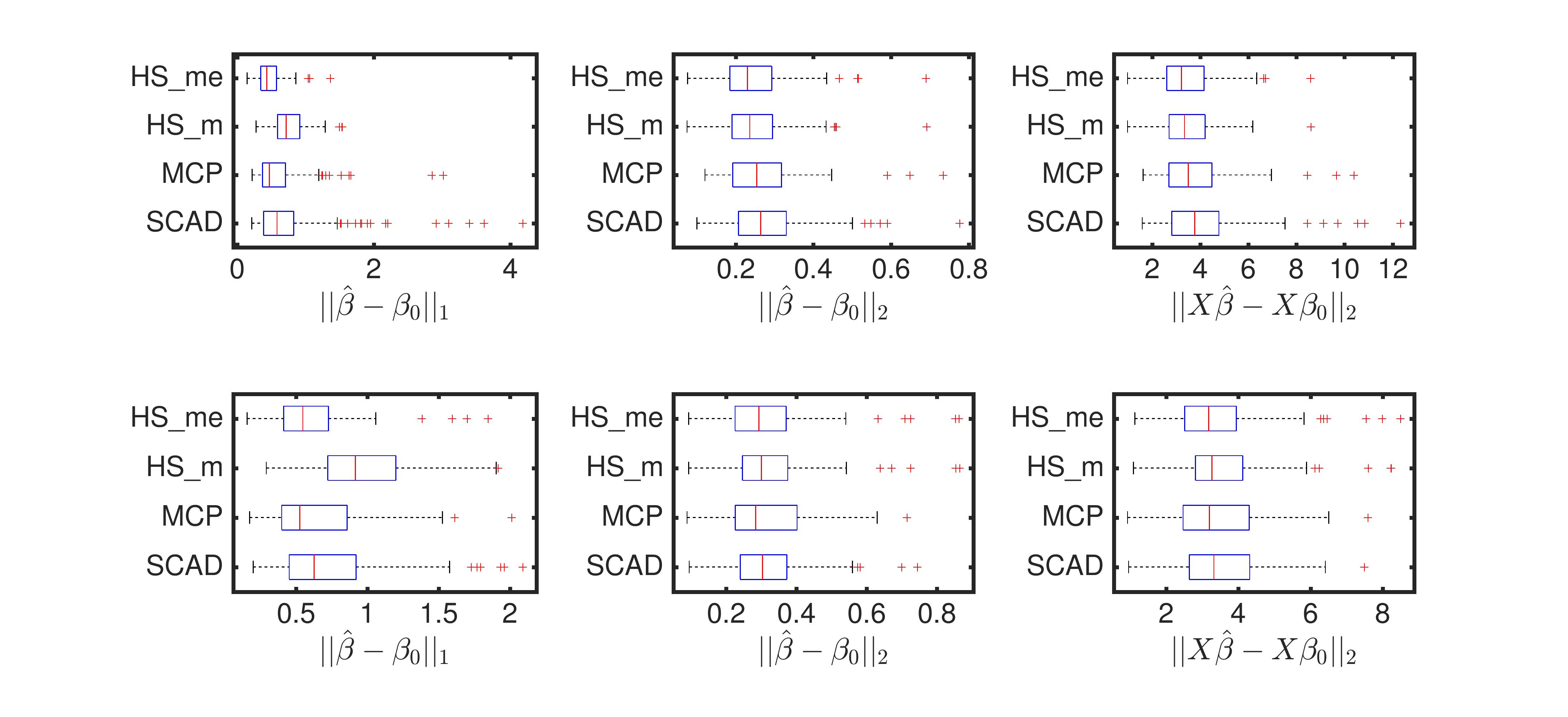}
\end{flushleft}
\label{fig:comp_1}
\end{figure}

\begin{figure}
\caption{Same setting as in Fig \ref{fig:comp_1}. True $\beta_0$ is $5$-sparse with non-zero entries $\pm \{0.75,1,1.25,1.5,1.75\}$.}
\begin{flushleft}
\includegraphics[width=1\textwidth]{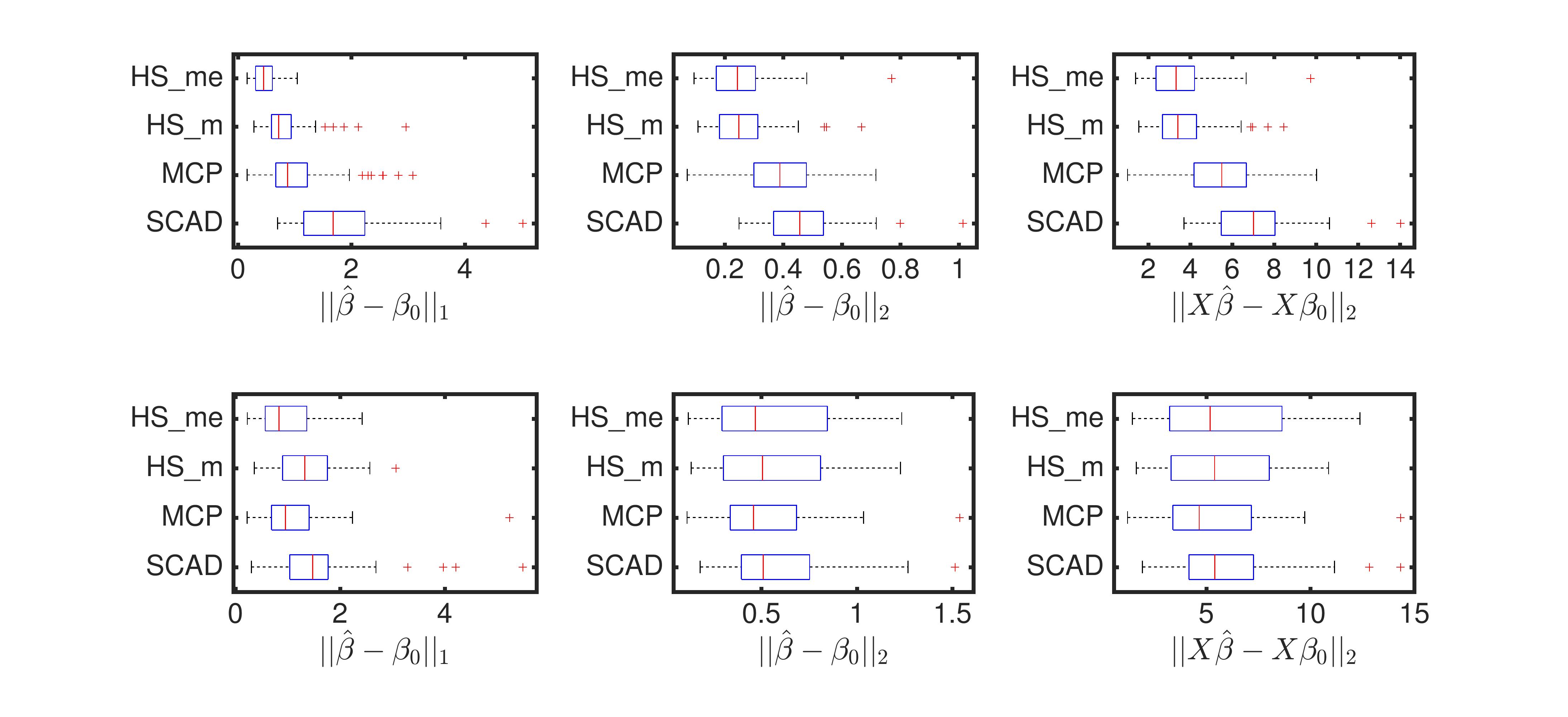}
\end{flushleft}
\label{fig:comp_2}
\end{figure}

We first report comparisons with smoothly clipped absolute deviation \citep{fan2001variable} and minimax concave penalty \citep{zhang2010nearly} methods. We considered model \eqref{eq:lin_reg} with $n = 200$, $p = 5000$ and $\sigma = 1.5$. Letting $x_i$ denote the $i$th row of $X$, the $x_i$s were independently generated from $\mbox{N}_p(0, \Sigma)$, with (i) $\Sigma = \mathrm{I}_p$ and (ii) $\Sigma_{jj}=1,\Sigma_{jj'}=0.5, j\ne j' =1,\ldots,p$, compound symmetry. The true $\beta_0$ had $5$ non-zero entries in all cases, with the non-zero entries having magnitude (a) $\{1.5,1.75,2,2.25,2.5\}$ and (b) $\{0.75,1,1.25,1.5,1.75\}$, multiplied by a random sign. For each case, we considered $100$ simulation replicates. The frequentist penalization approaches were implemented using the \texttt{R} package \texttt{ncvreg} via 10-fold cross-validation. 
For the horseshoe prior, we considered the posterior mean and the point-wise posterior median as point estimates. Figures 1 and 2 report boxplots for $\ell_1$, $||\hat{\beta} - \beta_0||_1$, $\ell_2$, $||\hat{\beta} - \beta_0||_2$, and prediction, $||X \hat{\beta} - X \beta_0||_2$, errors across the 100 replicates for the two signal strengths. The horseshoe prior is highly competitive across all simulation settings, in particular when the signal strength is weaker. An interesting observation is the somewhat superior performance of the point-wise median even under an $\ell_2$ loss; a similar fact has been observed about point mass mixture priors \citep{castilloneedles} in high dimensions. We repeated the simulation with $p = 2500$ with similar conclusions. Overall, out of the 24 settings, the horseshoe prior had the best average performance over the simulation replicates in 22 cases.

While there is now a huge literature on penalized point estimation, uncertainty characterization in $p > n$ settings has received attention only recently \citep{zhang2014confidence,van2014asymptotically,javanmard2014confidence}. Although Bayesian procedures provide an automatic characterization of uncertainty, the resulting credible intervals may not possess the correct frequentist coverage in nonparametric/high-dimensional problems \citep{szabo2015frequentist}. This led us to investigate the frequentist coverage of shrinkage priors in $p > n$ settings; it is trivial to obtain element-wise credible intervals for the $\beta_j$s from the posterior samples. We compared the horseshoe prior with \cite{van2014asymptotically} and \cite{javanmard2014confidence}, which can be used to obtain asymptotically optimal element wise confidence intervals for the $\beta_j$s. We considered a similar simulation scenario as before. We let $p \in \{500, 1000\}$, and considered a Toeplitz structure, $\Sigma_{jj'} = 0.9^{|j-j'|}$, for the covariate design \citep{van2014asymptotically} in addition to the independent and compound symmetry cases stated already. The first two rows of Table \ref{tab:cov} report the average coverage percentages and $100 \times$lengths of confidence intervals over $100$ simulation replicates, averaged over the $5$ signal variables. The last two rows report the same averaged over the $(p-5)$ noise variables. 

Table \ref{tab:cov} shows that the horseshoe has a superior performance. An attractive adaptive property of shrinkage priors emerges, where the lengths of the intervals automatically adapt between the signal and noise variables, maintaining the nominal coverage. The frequentist procedures seem to yield approximately equal sized intervals for the signals and noise variables. The default choice of the tuning parameter $\lambda \asymp (\log p/n)^{1/2}$ suggested in \cite{van2014asymptotically} seemed to provide substantially poorer coverage for the signal variables at the cost of improved coverage for the noise, and substantial tuning was required to arrive at the coverage probabilities reported. The default approach of \cite{javanmard2014confidence} produced better coverages for the signals compared to  \cite{van2014asymptotically}. The horseshoe and other shrinkage priors on the other hand are free of tuning parameters. The same procedure used for estimation automatically provides valid frequentist uncertainty characterization.

\begin{table}[!h]
\caption{Frequentist coverages (\%) and $100 \times $lengths of point wise $95\%$ intervals. Average coverages and lengths are reported after averaging across all signal variables (rows 1 and 2) and noise variables (rows 3 and 4). Subscripts denote $100 \times $standard errors for coverages. LASSO and SS respectively stand for the methods in \cite{van2014asymptotically} and \cite{javanmard2014confidence}. The intervals for the horseshoe (HS) are the symmetric posterior credible intervals.}
\huge
\begin{flushleft} 
\scalebox{0.34}{
\begin{tabular}{ccccccccccccccccccc}\toprule
{p} & \multicolumn{9}{c}{500} & \multicolumn{9}{c}{1000}\\
\cmidrule(lr){1-1} \cmidrule(lr){2-10}  \cmidrule(lr){10-19} \\
{Design} &\multicolumn{3}{c}{Independent} & \multicolumn{3}{c}{Comp Symm} & \multicolumn{3}{c}{Toeplitz} & \multicolumn{3}{c}{Independent} & \multicolumn{3}{c}{Comp Symm} & \multicolumn{3}{c}{Toeplitz} \\
\cmidrule(lr){1-1} \cmidrule(lr){2-4}  \cmidrule(lr){5-7}  \cmidrule(lr){8-10}  \cmidrule(lr){11-13}  \cmidrule(lr){14-16}  \cmidrule(lr){17-19} \\
& \HS & \LS & \SS & \HS & \LS & \SS & \HS & \LS & \SS & \HS & \LS & \SS & \HS & \LS & \SS & \HS & \LS & \SS\\
\cmidrule(lr){1-1} \cmidrule(lr){2-19}\\
Signal Coverage & $93_{1.0}$ & $75_{12.0}$ &  $82_{3.7}$ & $95_{0.9}$ & $73_{4.0}$ & $80_{4.0}$ & $94_{4.0}$ & $80_{7.0}$ & $79_{5.6}$ & $94_{2.0}$ & $78_{12.0}$ & $85_{5.1}$ &$94_{1.0}$ & $77_{2.0}$ & $82_{7.4}$ & $95_{1.0}$ & $76_{3.0}$ & $80_{8.3}$\\ \hline
Signal Length & $42$ & $46$ & $41$ & $85$ & $71$ & $75$ & $86$ & $79$ & $74$ & $39$ & $41$ & $42$ & $82$ & $76$ & $77$ & $105$ & $96$ & $95$\\ \hline 

Noise Coverage & $100_{0.0}$ & $99_{0.8}$ & $99_{1.0}$ & $100_{0.0}$ & $98_{1.0}$ & $99_{0.8}$ & $98_{1}$ & $98_{1.0}$ & $99_{0.6}$ & $99_{0.0}$ & $99_{1.0}$ & $98_{0.9}$ & $100_{0.0}$ & $99_{1.0}$ & $99_{0.1}$ & $100_{0.0}$ & $99_{1.0}$ & $99_{0.2}$\\ \hline
Noise Length & $2$ & $43$ & $40$ & $4$ & $69$ & $73$ & $5$ & $78$ & $73$ & $0.6$ & $042$ & $41$ & $0.7$ & $76$ & $77$ & $0.3$ & $98$ & $93$\\ 
\bottomrule
\end{tabular}
}  
\end{flushleft}
\label{tab:cov}
\end{table}

\section{Discussion}

Our numerical results warrant additional numerical and theoretical investigations into properties of shrinkage priors in high dimensions. The proposed algorithm can be used for essentially all the shrinkage priors in the literature and should prove useful in an exhaustive comparison of existing priors. Its scope extends well beyond linear regression. For example, extensions to logistic and probit regression are immediate using standard data augmentation tricks \citep{albert1993bayesian,holmes2006bayesian}. Multivariate regression problems where one has a matrix of regression coefficients can be handled by block updating the vectorized coefficient matrix ; even if $p < n$, the number of regression coefficients may be large if the dimension of the response if moderate. Shrinkage priors have been used as a prior of factor loadings in \cite{bhattacharya2011sparse}. While \cite{bhattacharya2011sparse} update the $p > n$ rows of the factor loadings independently, exploiting the assumption of independence in the idiyosyncratic components, their algorithm does not  
extend to approximate factor models, where the idiyosyncratic errors are dependent. The proposed algorithm can be adapted to such situations by block updating the vectorized loadings. Finally, we envision applications in high-dimensional additive models where each of a large number of functions is expanded in a basis, and the basis coefficients are updated in a block. 



\appendix
\section*{Appendix}

Here we give a constructive argument which leads to the poroposed algorithm. By the Sherman--Morrison--Woodbury formula (Hager, 1989) and some algebra we have, 
{\small \begin{align}\label{eq:wood}
\Sigma = (\Phi^{\T}\Phi+D^{-1})^{-1} = D -D \Phi^{\T}(\Phi D \Phi^{\T}+\mr I_n)^{-1}\Phi D, \quad \mu = D \Phi^{\T}(\Phi D \Phi^{\T}+\mr I_n)^{-1}\alpha. 
\end{align}}
However, the above identity for $\Sigma$ on its own does not immediately help us to sample from $\mbox{N}(0, \Sigma)$. Here is where the data augmentation idea is useful. Define $\zeta = (v^{\T}, u^{\T})^{\T} \in \Re^{n+p}$, where $u$ and $v$ are defined in step (i) and (ii) of the algorithm. Clearly, $\zeta$ has a mean zero multivariate normal distribution with covariance $
\Omega  = \left(\begin{matrix}
      P & S\\
      S^{\T} & R\\  
     \end{matrix}\right)
$, where $P = (\Phi D \Phi^{\T} + \mr I_n), S = \Phi D$ and $R = D$. The following identity is easily verified:
\begin{equation}\label{eq:mat_id}
\underbrace{ \left(\begin{matrix}
      P & S\\
      S^{\T} & R\\  
     \end{matrix}\right) }_{\Omega}
=
\underbrace{ \left(\begin{matrix}
      \mr I_n & 0\\
       S^{\T}P^{-1} & \mr I_p \\
     \end{matrix}\right)  }_{L}
\underbrace{ \left(\begin{matrix}
      P & 0\\
      0 & R-S^{\T}P^{-1}S \\    
     \end{matrix}\right) }_{\Gamma}
\underbrace{ \left(\begin{matrix}
      \mr I_n & P^{-1}S\\
      0 & \mr I_p \\   
     \end{matrix}\right) }_{L^{\T}}.
\end{equation}
Note that $\Gamma$ is block diagonal, with the lower $p \times p$ block of $\Gamma$ given by $R-S^{\T}P^{-1}S$ equaling $\Sigma$ by \eqref{eq:wood}. Further, $L$ is invertible, with $L^{-1}=\left( \begin{matrix} \mr I_n 
 & 0 \\ -S^{\T}P^{-1} & \mr I_p \\ \end{matrix} \right)$ which is easily derived since $L$ is lower triangular. Therefore, $\Gamma  = L^{-1} \Omega 
(L^{-1})^{\T}$. 

We already have $\zeta$ which is a sample from $\mbox{N}(0, \Omega)$. Defining $\zeta_* = L^{-1} \zeta$, clearly $\zeta_* \sim \Gauss(0, \Gamma)$. Thus if we collect the last $p$ entries of $\zeta^*$, they give us a sample from $\mbox{N}(0, \Sigma)$. Some further algebra produces the algorithm.

\bibliography{shrinkage_refs}
\bibliographystyle{biometrika}

\end{document}